\documentclass[10pt,a4paper]{amsart}
\usepackage[left=4cm,right=4cm,top=4cm,bottom=4cm]{geometry}
\usepackage{enumerate}
\usepackage{leftidx}
\usepackage{mathtools}
\usepackage{amsmath,amscd,amssymb,amsthm}
\usepackage{bm}
\usepackage{comment}
\usepackage[arrow,matrix]{xy}
\usepackage[OT2,T1]{fontenc}
\usepackage{booktabs}
\usepackage[position=bottom]{subfig}
\usepackage{xfrac, faktor}
\usepackage{lipsum}
\usepackage{graphicx}
\usepackage{setspace,color}
\usepackage{hyperref}
\numberwithin{equation}{section}
\numberwithin{table}{section}
\usepackage{array}
\usepackage{multirow}
\allowdisplaybreaks

\DeclareMathOperator{\rk}{rk}

\DeclareMathOperator{\Ker}{Ker}
\DeclareMathOperator{\enc}{enc}
\DeclareMathOperator{\ord}{ord}
\DeclareMathOperator{\height}{ht}

\theoremstyle{definition}
\newtheorem{definition}{Definition}[section]
\newtheorem{example}[definition]{Example}

\newtheorem*{remark*}{Remark}

\theoremstyle{plain}
\newtheorem{theorem}[definition]{Theorem}
\newtheorem{corollary}[definition]{Corollary}
\newtheorem{lemma}[definition]{Lemma}
\newtheorem{proposition}[definition]{Proposition}

\newcommand{\vF}{{ \mathbb F }}

\title{Optimal Rank-Metric Codes with Rank-Locality from Drinfeld Modules}

\author[L. Bastioni]{Luca Bastioni}
\address{University of South Florida \\ 4202 E Fowler Ave\\
33620 Tampa, US.}
\email{lbastioni@usf.edu}

\author[M. O. Darwish]{Mohamed O. Darwish}
\address{University of South Florida \\ 4202 E Fowler Ave\\
33620 Tampa, US.}
\email{moh13@usf.edu}

\author[G. Micheli]{Giacomo Micheli}
\address{University of South Florida\\ 4202 E Fowler Ave\\
33620 Tampa, US.}
\email{gmicheli@usf.edu}


\makeatletter
\@namedef{subjclassname@2020}{%
  \textup{2020} Mathematics Subject Classification}
\makeatother
\begin{document}

\maketitle


\begin{abstract}
We introduce a new technique to construct rank-metric codes using the arithmetic theory of Drinfeld modules, and Dirichlet Theorem on polynomial arithmetic progressions. Using our methods, we obtain a new infinite family of optimal rank-metric codes with rank-locality, i.e.\ every code in our family achieves the information theoretical bound for rank-metric codes with rank-locality.
\end{abstract}

\section{Introduction}
Let $q$ be a prime power and $m$ a positive integer.
Rank-metric codes are an important family of error correcting codes that are used for reliable communication over noisy channels \cite{alfarano2022linear,bartz2022rank,gorla2021rank,neri2022twisted,silva2008rank} and for cryptography \cite{couvreur2020hardness,gaborit2014new,samardjiska2019reaction}. They have been widely studied in the literature due to their interesting algebraic combinatorial properties that allow to encode information in a way that the obtained codewords are resilient to rank errors, i.e. the noise is modeled as a sum of a low rank matrix.

Locality is a foundational concept in information theory \cite{silberstein2013optimal,tamo2014family,tamo2016bounds}. It allows for efficient recovery of encoded information that is partially lost in distributed storage systems. More specifically, if the information vector is encoded and distributed over multiple servers, locality allows to recover the information contained in a single server (that might be unavailable for various reasons) using only a few other servers.

Rank-metric codes with rank-locality belong to an important subclass of rank-metric codes that allow efficient corrections of criss-cross failures in data centers (for additional detail on this see for example \cite[Example 1]{rankmetriccodes}).
Various constructions with flexible parameters are available for locally recoverable codes over the Hamming metric \cite{barg2017locally,barg2017locally2,chen2022function,dukes2023optimal,
garrison2023class,liu2020constructions,mesnager2019good,micheli2019constructions} but only one for the rank-metric codes \cite{rankmetriccodes}, which allows only very restrictive parameters (see for example the divisibility conditions in \cite[Construction 1]{rankmetriccodes}).
The difficulty that arises in the rank-metric context is that classical constructions fail due to both the non-commutative nature of the message space (which naturally consists of $q$-linearized polynomials with coefficients in $\vF_{q^m}$), and the fact that one can control certain splitting conditions only for special linearized polynomials. For example, in \cite[Construction 1]{rankmetriccodes}, the kernels of the linearized polynomials used are orthogonal subspaces whose direct sum is a field, which gives restrictive divisibility conditions.

The purpose of this paper is to employ the theory of Drinfeld Modules to construct new rank metric codes with locality having parameters that are not available using other constructions.
Our method employs a blend of supersingular Drinfeld Modules, Anderson motive decompositions, Dirichlet Theorem, and of course rank metric codes theory. Notice that differently from \cite{silberstein2013optimal} (where the authors get classical locality, as explained in \cite[Section 3]{rankmetriccodes}), we actually get proper rank-locality, as required in \cite{rankmetriccodes}. 

We structure the paper in three parts. The first one, Section \ref{sec:prel}, is merely introductory and provides the necessary background on Drinfeld Modules and Rank Metric Codes.
The second one, consisting of Sections \ref{sec:genconstr} and \ref{sec:TwoCases}, provides the algebraic core of the construction. Then, in the third part we show that the hypothesis of our theorem can indeed be easily verified. In order to see this, we employ Dirichlet Theorem for polynomial arithmetic progressions to show the existence of an infinite family of our codes (Section \ref{sec:infinfam}), and provide two examples to show that the parameters can be kept small (Section \ref{sec: examples}).

\section{Preliminaries}\label{sec:prel}

\subsection{Drinfeld Modules}\label{sec: DM}
In this subsection, we provide some background on the theory of Drinfeld modules. For a comprehensive treatise of the subject, the reader should refer to \cite{Papikian2023}. Throughout this section we assume $K\supseteq \vF_q$ is a field and $a\in \vF_q[T]$.
Throughout the article, let $q$ be a prime power and $\vF_q$ the finite field with $q$ elements. We denote by $A=\vF_q[T]$ the ring of polynomials over $\vF_q$.

\begin{definition}
Let $\vF_{q^n}$ be a finite extension of $\vF_q$. We define the norm 
\begin{align*}
  \operatorname{Nr}: \vF_{q^n}&\longrightarrow    \vF_{q^n}  \\
     \alpha & \longmapsto \prod_{i=0}^{n-1} \alpha^{q^i} \\
\end{align*}
\end{definition}

\begin{definition}
 Let $P\in \vF_q[T]$ be an irreducible polynomial and $\mathfrak{p}=(P)$ be the ideal generated by $P$. We say $\vF_q[T]/\mathfrak{p}$ is the \textit{residue field} of $\mathfrak{p}$ and we shall denote it by $\vF_\mathfrak{p}$.      
\end{definition}

Slightly abusing the notation, we also denote by $\mathfrak p$ the monic generator of $\mathfrak p=(P)$.

\begin{definition}
Let $n\geq 0$. A polynomial $f(x)=\sum_{i=0}^na_ix^{q^i}$ in $K[x]$ is called a $q$-\emph{linearized} polynomial. If $\operatorname{deg}(f)=q^n$, we say that $f$ has $q$-degree equal to $n$ and we denote $q^i$ by $[i]$. The ring of $q$-linearized polynomials is denoted by $K\langle x \rangle$, and in general, it is a noncommutative ring under addition and composition.      
\end{definition}

\begin{definition}
Let $K$ be a commutative $\vF_q$-algebra and $\tau$ be an indeterminate.  Let $K\{\tau\}$ be the set of polynomials of the form $\sum_{i=0}^na_i\tau^i$ such that $a_i\in K$ and $n\geq 0$. In $K\{\tau\}$, for $\{a_i,b_i\}_{i=0,\ldots,n}\in K$, we define the addition as the usual addition of polynomials
$$\sum_{i=0}^na_i\tau^i+\sum_{i=0}^nb_i\tau^i=\sum_{i=0}^n(a_i+b_i)\tau^i$$ while the multiplication is defined as follows for any $a,b\in K$
\begin{equation}\label{twistedmultiplication}
(a\tau^i)(b\tau^j):=ab^{q^i}\tau^{i+j}.    
\end{equation}
With these operations $K\{\tau\}$ is a noncommutative ring called the ring of \textit{twisted} polynomials.
\end{definition}
\begin{definition}
Let $f = a_h\tau^h + a_{h+1}\tau^{h+1} \dots + a_n\tau^n \in K\{\tau\}$ such that $0 \leq h \leq n$ with $a_h \neq 0, a_n \neq 0$. The \emph{height} of $f$ is $\height(f) = h$, and the $\tau$-\emph{degree} (or simply the \emph{degree}) of $f$ is $\deg(f)\coloneqq\deg_\tau(f) = n$. We also say that $f$ is \emph{separable} if $\height(f)=0$ and \emph{inseparable} otherwise.
\end{definition}

\begin{definition}\label{drinfeldmodule}
Let $K$ be an $\vF_q[T]$-field, (i.e. a field equipped with an $\vF_q$-algebra homomorphism $\gamma:\vF_q[T]\rightarrow K$).
A Drinfeld module of rank $R\geq 1$ is an $\vF_q$-algebra homomorphism 
\begin{align*}
\phi:\vF_q[T]\longrightarrow &K\{\tau\}\\
a\longmapsto &\phi_a=\gamma(a)+\sum^n_{i=1}g_i(a)\tau^i
\end{align*}
where $g_n(a)\neq 0$ and $n=\operatorname{deg}_T(a)\cdot R$. The morphism $\gamma$ is called the \emph{structure morphism} of $\phi$, and $\mathfrak{p} \coloneqq \ker(\gamma)$ is called the $A$-\emph{characteristic} of $K$ and denoted by $\operatorname{char}_A(K)$.
\end{definition}

For what follows, it is important to observe that $K$ could have different structures of $A$-field, namely it could be associated to different $A$-characteristics, depending on the choice of $\gamma$ and $\phi$.
Observe also that for any  $a=\sum_{i=0}^na_iT^i$, we have $\phi_a=\sum_{i=0}^na_i\phi_T^i$. In other words, $\phi_T$ determines $\phi$ entirely.
Moreover, $K\{\tau\}$ is related to the ring of $q$-linearized polynomials $K\langle x \rangle$ via the following isomorphism 
 \begin{align*}
\iota: K\{\tau \}&\longrightarrow K\langle x \rangle\\ 
a_i\tau^i&\longmapsto a_ix^{q^i}.
 \end{align*}      

\begin{definition}
$\iota(\phi_a)$ is called the $a$-\emph{division polynomial} of $\phi$ and denoted by $\phi_a(x)$.     
\end{definition}

Among all Drinfeld modules, there is a particular one that plays a special role in our paper. Therefore we introduce also the following.
\begin{definition}\cite[Definition 3.2.5]{Papikian2023}
The \textit{Carlitz module} $C$ is the Drinfeld module of rank $1$ defined by $C_T=\gamma (T)+\tau$.
\end{definition}

It is a well-known fact in finite field theory that a $q$-linearized polynomial in $K\langle x \rangle$ is also an $\vF_q$-linear operator on $K$ (see \cite[Chapter 3.1]{Papikian2023}). Therefore, one can talk about $\Ker\left(\phi_a(x)\right)\subseteq \overline K$.

\begin{definition}
An element $\alpha \in \operatorname{Ker} \left( \phi_a(x) \right) \subseteq \overline{K}$ is called an \emph{$a$-torsion point}. The set of all $a$-torsion points is denoted by $\phi[a]$. In addition, we denote by $K(\phi[a])$ the splitting field of $\phi_a(x)$ over $K$ and refer to it as the \emph{$a$-division field} of $\phi$.
\end{definition}

\begin{proposition}\label{prop:height}\cite[Lemma 3.2.11]{Papikian2023}
If $\operatorname{char}_A(K) = \mathfrak{p} \neq 0$ and $\phi$ is a Drinfeld module of rank $R$ over $K$, then there is an integer $1 \leq H (\phi) \leq R$, called the \textbf{height} of $\phi$, such that for all $0\neq a\in A$ we have
\[
\height(\phi_a) = H(\phi)\cdot \ord_\mathfrak{p}(a)\cdot \deg_T(\mathfrak{p})
\]
where $\ord_\mathfrak{p}(a)$ is the largest power of $\mathfrak{p}$ dividing $a$.     
\end{proposition}

If a Drinfeld module $\phi$ of rank $R$ over a finite field field is such that $H(\phi)=R$, then $\phi$ is called \emph{supersingular}. The interested reader should consult \cite[Section 4.4]{Papikian2023} for further background and additional properties of supersingular Drinfeld modules. Such type of Drinfeld modules are quite special, as the following result shows. 

\begin{proposition}\label{prop:SupersingEquiv}
Let $K\supseteq \vF_q$ be a finite field, and $\phi:\vF_q[T]\rightarrow K\{\tau\}$ be a Drinfeld module of rank $R$ and $\operatorname{char}_A(K)= \mathfrak{p}$. The followings are equivalent:
\begin{enumerate}
    \item $\phi$ is supersingular (i.e. $H(\phi)=R)$
    \item $\phi[\mathfrak{p}^k]=\{0\}$ for all $k\geq 1$ 
\end{enumerate}
  
\end{proposition}
\begin{proof}
Refer to \cite[Theorem 4.4.1]{Papikian2023} for a proof and more equivalent conditions for a Drinfeld module to be supersingular. 
\end{proof}

By definition of height, the simplest example of supersingular Drinfeld module is given by any Drinfeld module $\phi$ of rank 1, indeed $1\leq H(\phi)\leq R=1$. Another example is shown below.

\begin{example}\label{ex:supersing}\cite[Example 4.4.5]{Papikian2023}
Let $R$ be a prime, and $m$ be an integer coprime to $R$. Let $\alpha$ be a root of an irreducible polynomial $\mathfrak{p}\in \vF_q[T]$ of degree $m$. The Drinfeld module $\phi:\vF_q[T]\rightarrow \vF_{q^m}\{\tau\}$ given by $\phi_T=\alpha +\tau^R$ is supersingular.     
\end{example}

It turns out that the set of $a$-torsion points has a more general structure than just being an $\vF_q$-vector space. The following result tells that $\phi[a]$ has an $A$-module structure, or equivalently that every $f\in\phi(A)$ acts as an endomorphism of $\phi[a]$.

\begin{proposition}
For any fixed $a\in\vF_q[T]$, the set of $a$-torsion points $\phi[a]$ is a finite $\vF_q[T]-$module with scalar multiplication defined by $b\circ \alpha=\phi_b(\alpha)$, for some $b\in \vF_q[T]$ and $\alpha \in \phi[a]$.
\end{proposition}
\begin{proof}
Just notice that $\phi_a(\phi_b(\alpha))=\phi_b(\phi_a(\alpha))=\phi_b(0)=0$.
\end{proof}

Another useful fact about the structure of the $ab$-torsion points for coprime $a,b\in \vF_q[T]$ is the following proposition. 

\begin{proposition}\cite[Lemma 3.5.1]{Papikian2023}\label{prop:orthogonality}
If $a,b \in \vF_q[T]$ are relatively prime nonzero elements, then, $\phi[ab]=\phi[a]\times \phi[b]$ as $\vF_q[T]$-modules.   
\end{proposition}

\subsection{Rank Metric Codes}\label{sec:rkmetricbg}
In this section, we provide some background about rank-metric codes based on \cite[Ch II.11]{huffman2021concise}. Throughout this section, we denote by $\vF^{m\times n}_q$ the $m\times n$ matrices with entries in  $\vF_q$. Now, we define the rank distance on $\vF^{m\times n}_q$, which is of fundamental importance in the theory of rank-metric codes.
\begin{definition}
Let $A,B\in \vF^{m\times n}_q$ and let $\rk_q(A)$ be the usual rank of $A$. We define the \textit{rank distance} between $A$ and $B$ as $d_R(A,B)=\rk_q(A-B)$. 
\end{definition}
One can check that this is indeed a distance on $\vF^{m\times n}_q$.
\begin{definition}
A (matrix) \textit{rank-metric code} $\mathcal{C}$ is an $\vF_q-$ linear subspace of $\vF^{m\times n}_q$. We define the $\textit{minimum distance}$ or simply the distance of a nonzero code to be
$$d_R(\mathcal{C}):=\min \{d_R(A,B): A,B\in \mathcal{C} \textup{ and } A\neq B\}$$
\end{definition}
For any $x=(x_1,\ldots,x_n)\in \vF^n_{q^m}$, we can define a matrix $M_x\in \vF^{m\times n}_q$ associated with $x$ up to the choice of the basis as follows. Fix a basis $\mathcal{B}=\{b_1,\ldots,b_m\}$ of $\vF_{q^m}$ over $\vF_q$. Then, $x_j=\sum_{i=1}^mc_{ij}b_i$ for some $c_{ij}\in \vF_q$. Thus, $M_x=[c_{ij}]_{\substack{1\leq i\leq m\\ 1\leq j\leq n}}$. In conclusion, we can define another type of rank-metric codes; namely (vector) rank-metric codes.  
\begin{definition}
	A (vector) \textit{rank-metric} \textit{code} $\mathcal C$ is an $\vF_{q^m}$-subspace of  $\vF^n_{q^m}$. The \textit{minimum distance} of a nonzero code $\mathcal C$ is $d_R(\mathcal C):=\min\{\rk_q(M_x):x\in\mathcal C\textup{ and } x\neq 0\}$. 
\end{definition}
It is immediate to see that every (vector) rank-metric code in $\vF^n_{q^m}$ can be seen as a (matrix) rank-metric code in $\vF_q^{m\times n}$ by fixing an $\vF_q$-basis for $\vF_{q^m}$. Note that our main construction is a (vector) rank-metric and hence also a (matrix) rank-metric code. Next, we introduce the Singleton bound for rank-metric codes.   
\begin{proposition}\cite[Theorem 11.3.5, Remark 11.3.6]{huffman2021concise}\label{prop: singletonRankMetric}
 Let $\mathcal{C}\subseteq \vF_q^{m\times n}$. Every  rank-metric code with minimum distance $d_R(\mathcal{C})$ satisfies the following upper bound.   
 $$|\mathcal{C}|\leq q^{\max\{m,n\}\left(\min\{m,n\}-d_R(\mathcal{C})+1\right)}.$$
In particular, if $\mathcal{C}\subseteq \vF_{q^m}^n$ is a vector rank-metric code, then, 
 \begin{equation*}
     \dim_{\vF_{q^m}}(\mathcal{C})\leq n-d_R(\mathcal{C})+1.
 \end{equation*}
 \end{proposition}
 \begin{definition}
     If a rank-metric code achieves equality in the previous upper bound, we call it a \textit{maximum rank-distance code (MRD)}.
 \end{definition}
 From now on, all our rank-metric codes will be $\vF_q$-subspaces.
 An important family of MRD codes is the family of Gabidulin codes (see \cite{Gabidulin}). We will need properties of these codes in what follows.
\begin{definition}\label{Gabidulin}
Let $n,m,k$ be positive integers such that $m\geq n$. Let $\mathcal{P} = \{p_1,\ldots, p_n\}$ be a set of elements
in $\vF_{q^m}$ that are $\vF_q$-linearly independent. Let $\boldsymbol{m}=(\boldsymbol{m}_0,\ldots,\boldsymbol{m}_{k-1})\in \vF^k_{q^m}$ and 
$\{G_{\boldsymbol{m}}(x)=\sum_{j=0}^{k-1}\boldsymbol{m}_jx^{[j]}\}$ be the set of linearized polynomials of $q$-degree at most $k-1$ in  $\vF_{q^m}[x]$. We define the encoding map as \begin{align*}
 \enc : \vF^k_{q^m}\longrightarrow &\vF^n_{q^m}\\
  \boldsymbol{m}\longmapsto&\{G_{\boldsymbol{m}}(\gamma),\gamma \in \mathcal{P}\}.
\end{align*} 
An $(n,k)$-Gabidulin Code $\mathcal{C}_{\operatorname{Gab}}$ over $\vF_{q^m}$ is defined as $$\mathcal{C}_{\operatorname{Gab}}=\{G_{\boldsymbol{m}}(\gamma): \gamma \in \mathcal{P}, \boldsymbol{m}\in \vF^k_{q^m}\}.$$

\end{definition}

\begin{definition}\cite[Definition 3]{rankmetriccodes}\label{ranklocality}
We say that $\mathcal C\subseteq \vF_q^{m\times n}$ has rank-locality $(r,\delta)$ if, for every column index $i\in \{1,\dots,n\}$, 
there exists a set of columns $\Gamma(i)\subseteq \{1,\dots,n\}$ such that
\begin{itemize}
\item $i\in \Gamma(i)$
\item $|\Gamma(i)|\leq r+\delta -1$
\item $d_{R}(\mathcal C \mid_{\Gamma(i)})\geq \delta$
\end{itemize}
where $\mathcal C\mid_{\Gamma(i)}$ is the code obtained from all the (matrix) codewords in $\mathcal C$ restricted to the columns indexed by $\Gamma(i)$. The code $\mathcal C\mid_{\Gamma(i)}$ is said to be the local code associated with the $i$-th column. A rank-metric code $\mathcal C\subseteq \vF_q^{m\times n}$ of $\vF_q$-dimension $K$ is  said to be an $(m\times n,K)$-code. In particular, an $(n,k)$-Gabidulin code over $\vF_{q^m}$ is an $(m\times n, mk)$ rank-metric code over $\vF_q$. An $(m\times n,K)$-code  with minimum rank  distance $d_R$ and $(r,\delta)$ rank-locality is denoted as an $(m\times n,K,d_R,r,\delta)$ rank-metric code. 
\end{definition}

\begin{proposition}\label{prop: singletonLikeRankMetric}
\cite[Equation 9]{rankmetriccodes}
The rank distance $d_R(\mathcal{C})$ of an $(m\times n,mk)$ rank-metric code with $(r,\delta)$ rank-locality satisfies the following
$$d_R(\mathcal{C})\leq n-k+1-\left(\left\lceil \frac{k}{r}\right\rceil-1\right)(\delta-1)$$
\end{proposition}
\begin{proof}
See \cite[Theorem 1]{rankmetriccodes}    
\end{proof}

\section{General Theorem and Main Construction}\label{sec:GeneralMainConstruction}\label{sec:genconstr}
For this section and the theorems within, we will use the following notation so we will be able to make the statements of each theorem lighter.
Let $q$ be a prime power, and $\vF_q$ be the finite field with $q$ elements. For $m,t>0$ let $\vF_{q^m}$ be the algebraic field extension of $\vF_q$ of degree $m$, and $K$ be a finite extension of $\vF_{q^m}$ of degree $t$ (namely $[K:\vF_q]=tm)$. We denote with $A\coloneqq\vF_q[T]$ the ring of polynomials in $T$ with coefficient over $\vF_q$. Let $0<\ell<q$, $r\geq 1$ and $\delta\geq 2$ be positive integers, and define $R\coloneqq r+\delta-1$.

\begin{theorem}\label{thm:GeneralMainThm}
    Choose distinct $a_1,\ldots,a_\ell\in\vF_q^\times$ and define the polynomial $h=\prod_{i=1}^\ell(T-a_i)\in A$. Let $s$ be such that $s+1\leq\ell$. Let $\phi:A\rightarrow \vF_{q^m}\{\tau\}$ be a Drinfeld module of rank $R$, with $A$-characteristic being a monic irreducible polynomial $\mathfrak{p}$ of degree $m$ such that $\gcd(\mathfrak{p}, h)=1$. Assume $\vF_{q^m}(\phi[h])\subseteq K$. Then, there exists a
    \[
    \left(tm\times \ell R, tm(s+1)r,\ell R-Rs-r+1,r,\delta\right)
    \]
    (matrix) rank-metric code that is optimal w.r.t. the bound shown in Proposition~\ref{prop: singletonLikeRankMetric}.
\end{theorem}
In subsection \ref{sec:CodeConstruction} we will prove Theorem \ref{thm:GeneralMainThm}.
Knowing that $\vF_{q^m}(\phi[h])\subseteq K$ is fundamental to construct our code. In Section~\ref{sec:TwoCases} we will show the existence of at least two families of Drinfeld modules for which such assumption is always verified. 
This will amount to finding an irreducible polynomial  within certain wanted residue class modulo $h$. 
In general this is possible even for small $m$ (see Section~\ref{sec:exampleCarlitz} and \ref{sec:exampleSuper} for examples). In Section~\ref{sec:infinfam} we will use Dirichlet Theorem for polynomial arithmetic progressions to provide a proof  that there exists a wide range of parameters for which we can guarantee the existence of a monic irreducible $\mathfrak{p}$ such that $\vF_{q^m}(\phi[h])\subseteq K$.

\begin{lemma}
With notation and hypothesis as in Theorem~\ref{thm:GeneralMainThm}, we have that $$\dim_{\vF_q}\phi[h]=\ell R.$$ In particular $\dim_{\vF_q}\phi[T-a_i]=R.$
\end{lemma}
\begin{proof}
Since $\gcd(\mathfrak{p},h) =1$, by Proposition~\ref{prop:height}, we have
\[
\height(\phi_{h}) = H(\phi)\cdot\ord_\mathfrak{p}(h)\cdot\deg(\mathfrak{p})=0.
\]
This means that $\phi_{h}$ is separable, and since $\deg_\tau(\phi_h)=\ell R$, then $\dim_{\vF_q}\phi[h]=\dim_{\vF_q}\Ker(\phi_{h}(x))=\ell R$. In particular, since $\gcd(\mathfrak{p},h) =1$, it is also $\gcd(\mathfrak{p},T-a_i)=1$ for every $i=1,\ldots,\ell$. Therefore, applying the same argument we prove our claim.
\end{proof}

\subsection{Main Code Construction}\label{sec:CodeConstruction}
In this subsection we prove Theorem \ref{thm:GeneralMainThm}.
Let us start the construction of our code by choosing distinct $a_1,\dots a_\ell\in \vF_q^\times$ and $h=\prod^\ell_{i=1}(T-a_i)$ for some $\ell>0$. Let $\phi$ be a Drinfeld module of rank $R=r+\delta-1$ over $\vF_{q^m}$ with $\vF_q[T]$-characteristic being a monic irreducible polynomial $\mathfrak{p}$ of degree $m$ such that $\gcd(\mathfrak{p},h)=1$ and $\vF_{q^m}(\phi[h])\subseteq K=\vF_{q^{tm}}$, for some $t\geq 1$.

The following conditions are a consequence of the hypothesis of Theorem \ref{thm:GeneralMainThm}. We will be using them to prove the theorem.

\begin{enumerate}\label{conditions}
    \item[(C0)] $\ell<q$: since the $a_i$'s have to be distinct,
    \item[(C1)] $s+1\leq\ell $: included in the hypothesis, 
    \item[(C2)] $\ell R \leq tm$: notice that, since $h$ has all distinct roots and  $\gcd(\mathfrak p,h)=1$, then by Proposition \ref{prop:height} we have that $\phi_h$ has all distinct roots from which it follows $\dim_{\vF_q}(\phi[h])=\ell R\leq \dim_{\vF_q}(K)$.
\end{enumerate}
Roughly, these conditions simply allow for enough space in $\vF_q^{tm\times \ell R}$ to construct a code with the wanted dimension and $(r,\delta)$-rank-locality.

The message space $\mathcal M$ is given by 
\[
\mathcal M=\left\{\sum^{s}_{k=0} g_k(\tau) {\phi}_T^k \mid g_k\in \vF_{q^{tm}}\{\tau\}_{\leq r-1}\right\}.
\]
For $i=1,\ldots,\ell$ let $W_{i}=\phi[T-a_i]$ and observe that by Propostion \ref{prop:orthogonality} we have that $\phi[h]=\oplus^{\ell}_{i=1} W_i$. Let $\{\beta^{(i)}_{1},\dots, \beta^{(i)}_R\}$ be a basis for $W_{i}$. The encoding map is then given by
\begin{align*}
\enc: \mathcal M & \longrightarrow \vF^{\ell R}_{q^{tm}}\\*
f &\longmapsto \left(f(\beta^{(i)}_j)\right)_{\substack{i\in \{1,\dots,\ell\}\\ j\in \{1,\dots,R\}} }.
\end{align*}
C1 guarantees that the encoding is injective, and therefore one obtains the wanted dimension. Moreover, the fact that $\vF_{q^m}(\phi[h])\subseteq K$ gives that the $\beta_j^{(i)}$'s are independent elements of $K=\vF_{q^{tm}}$.

\begin{theorem}\label{thm:CodeOptimal}
The code defined by $\enc(\mathcal M)$ is an optimal
\[
\left(tm\times \ell R, tm(s+1)r,\ell R-Rs-r+1,r,\delta\right)
\]
(matrix) rank-metric code w.r.t. the bound in Proposition~\ref{prop: singletonLikeRankMetric}.
\end{theorem}
\begin{proof}
Let $W=\phi[h]=\phi\left[\prod_{i=1}^{\ell}(T-a_i)\right]$. To prove the length of the code is $\ell R$, first, we show that the $W_i$'s are orthogonal. Observe that $\phi[T-a_i]$ is an $\vF_q[T]-$module. Given $\operatorname{gcd}(T-a_1,\ldots,T-a_{\ell})=1$ then, by Proposition \ref{prop:orthogonality}, we have
\[
W=\phi\left[\prod_{i=1}^{\ell}(T-a_i)\right]=\prod_{i=1}^{\ell}\phi[T-a_i]=\prod_{i=1}^{\ell}W_i.
\]
Now, if we consider $W_i$'s as $\vF_q$-vector spaces, since $W_j\cap \bigoplus_{\scriptstyle \substack{i=1\\  i\neq j}}^{\ell}W_i=\{0\}$, then $W=\bigoplus_{i=1}^{\ell}W_i$.
In particular, $\{\beta_{j}^{(i)}\}_{\substack{1\leq i\leq \ell\\ 1\leq j\leq R}}\subseteq \vF_{q^{tm}}$ is an $\vF_q$-basis of $W$ and thus  $\dim_{\vF_q}(W)=|\{\beta_{j}^{(i)}\}_{\substack{1\leq i\leq \ell\\ 1\leq j\leq R}}|=\ell R$. This shows that the length of $\enc(\mathcal M)$ (as a vector rank-metric code) is $\ell R$. Finally, by fixing an $\vF_q$-basis for $\vF_{q^{tm}}$, each element in $\enc(\mathcal M)$ can be seen as a matrix in  $\vF_q^{tm\times \ell R}$.

We now show that the $\vF_{q^{tm}}$-dimension of the code is $(s+1)r$. Set $f(\tau)\in \mathcal{M}$. From the construction, $f(\tau)=\sum^{s}_{k=0} g_k(\tau) {\overline{\phi}}_{T}^k$ for some $g_k\in \vF_{q^{tm}}\{\tau\}_{\leq r-1}$. Now, for every $k\in \{0,\ldots,s\}$, we have $g_k(\tau)=\sum_{j=0}^{r-1}a^{(k)}_j\tau^j$ where $a^{(k)}_j\in \vF_{q^{tm}}$ for every $j\in\{0,\ldots,r-1\}$.
This means that every vector $$\boldsymbol{a}=\left(a^{(0)}_0,\ldots,a^{(0)}_{r-1},\ldots, a^{(k)}_0,\ldots,a^{(k)}_{r-1},\ldots,a^{(s)}_0,\ldots,a^{(s)}_{r-1}\right)\in \vF_{q^{tm}}^{(s+1)r}$$ can be associated to $g_0(\tau),\ldots,g_s(\tau)\in \vF_{q^{tm}}\{\tau\}_{\leq r-1}$, and thus to each $f(\tau)\in \mathcal{M}$. This shows that the $\vF_{q^{tm}}$-dimension of the code is $(s+1)r$ as long as the encoding map is injective. To see the injectivity of the map $\enc$, simply observe that $\deg_\tau(\phi_T)=R$, and a non-zero polynomial of linearized degree $sR+r-1<(s+1)R\leq \ell R$ cannot be the zero map on a space of $\vF_{q^{tm}}$-dimension $\ell R$. This shows that the $\vF_q$-dimension of the code is $tm(s+1)r$.

We now need to establish $(r,\delta)$-rank-locality according to Definition \ref{ranklocality}. Considering the code as a matrix rank-metric code (i.e. codewords are matrices), suppose $\gamma\in \{1,\ldots,\ell R\}$ is a column index. Then, the $\gamma$-th column in every codeword is an entry of the form $f(\beta^{(i_\gamma)}_{j_\gamma})$ for some $i_\gamma\in \{1,\ldots,\ell\}$ and $j_\gamma\in \{1,\ldots,R\}$ where $\beta^{(i_\gamma)}_{j_\gamma}$ is an element of the $\vF_q$-basis of $W_{i_\gamma}$.
If $\{\beta^{(i_\gamma)}_1,\ldots,\beta^{(i_\gamma)}_R\}$ is the $\vF_q$-basis of $W_{i_\gamma}$, then we define the set of column indexes $\Gamma(\gamma)\subset \{1,\ldots,\ell R\}$ to be the set of indexes of the columns associated to $f(\beta^{(i_\gamma)}_1),\ldots,f(\beta^{(i_\gamma)}_{R})$. It is obvious that $\gamma\in \Gamma(\gamma)$ and $|\Gamma(\gamma)|=\dim_{\vF_q}(W_{i_\gamma})=R=r+\delta-1$, satisfying the first two conditions of Definition~\ref{ranklocality}.

To check the third, consider the local code $\enc(\mathcal{M})|_{\Gamma(\gamma)}$ associated with the $\gamma$-th column. For a fixed $1\leq i_\gamma\leq \ell$ and basis $\{\beta_{j}^{(i_\gamma)}\}_{ 1\leq j\leq R}\subseteq \vF_{q^{tm}}$ of $W_{i_\gamma}=\phi[T-a_{i_\gamma}]$, the local code associated with the $\gamma$-th column is as follows: 
\[
\enc(\mathcal{M})|_{\Gamma(\gamma)}=\left \{\left(f(\beta^{(i_\gamma)}_1),\ldots,f(\beta^{(i_\gamma)}_R) \right) \bigg| f\in \mathcal{M} \right\}.
\]
However, we can rewrite $f(x)\in \vF_{q^{tm}}\langle x\rangle$ restricted to $\phi[T-a_{i_\gamma}]$ as
\begin{equation}\label{eq:restrict}
f\vert_{\phi[T-a_{i_\gamma}]}(x)=\sum_{k=0}^s\left(g_k( a_{i_\gamma}^kx)\right)=\sum_{k=0}^sa^k_{i_\gamma}g_k(x).
\end{equation}
where $g_k(x)\in \vF_{q^{tm}}\langle x\rangle$ with $q$-degree at most $r-1$. If $\vF_{q^{tm}}\langle x \rangle_{\leq r-1}$ denotes the set of $q$-linearized polynomials with coefficients in $\vF_{q^{tm}}$ and $q$-degree at most $r-1$, then we have that $f\vert_{\phi[T-a_{i_\gamma}]}(x)$ lives in $\vF_{q^{tm}}\langle x \rangle_{\leq r-1}$. In light of this, we observe that 
\[
\enc(\mathcal{M})|_{\Gamma(\gamma)}= \left \{\left(p(\beta^{(i_\gamma)}_1),\ldots,p(\beta^{(i_\gamma)}_R) \right) \bigg| p\in \vF_{q^{tm}}\langle x \rangle_{\leq r-1}\right\}=:G.
\]
By Definition \ref{Gabidulin}, $G$ lives in an $(R,r)$-Gabidulin code. Since Gabidulin codes are MRD and using  Proposition $\ref{prop: singletonRankMetric}$, $d_R(\enc(\mathcal{M})|_{\Gamma(\gamma)})=R-r+1=\delta$. This proves that $\enc (\mathcal{M})$ is an $\left(tm\times \ell R, tm(s+1)r\right)$ rank-metric code with $(r,\delta)$-rank-locality. In view of this and by Proposition \ref{prop: singletonLikeRankMetric}, we conclude that 
$$
d_R(\enc(\mathcal M))\leq \ell R - (s+1)r+1-(s+1-1)(R-r)=\ell R-Rs-r+1.$$
Finally, observe that  $\enc(\mathcal M)$ is the set of evaluations of $q$-polynomials of the form $\sum^{s}_{k=0} g_k(x)\circ \phi_T^k(x)$ with $q$-degree at most $Rs+r-1$ (maximum  $q$-degree of $g_k(x)$ and $\phi_T^k(x)$ are at most $r-1$ and $Rs$ respectively) over an $\vF_q$-linearly independent set of size $\ell R\leq tm$ which is $\{\beta_{j}^{(i)}\}_{\substack{1\leq i\leq \ell\\ 1\leq j\leq R}}$. In other words, $\enc(\mathcal M)$ is a subcode of an $(\ell R,Rs+r)$ Gabidulin code. Hence
$$d_R\left(\enc(\mathcal M)\right)\geq \ell R-Rs-r+1.$$
In conclusion, $\enc(\mathcal{M})$ is an optimal $\left(tm\times \ell R, tm(s+1)r,\ell R-Rs-r+1,r,\delta\right)$ matrix rank-metric code.
\end{proof}

\section{Supersingular Drinfeld Modules}\label{sec:TwoCases}

In this section, we show how to realize the ambient space $K$ of our code by choosing proper types of Drinfeld modules. Our goal is to describe enough conditions for the inclusion $\vF_{q^m}(\phi[h])\subseteq K$ to occur, which is critical to be able to work over extension fields whose degree can be controlled.

\subsection{Supersingular Drinfeld modules of rank $R>1$}\label{sec: supersingular}
In this subsection we focus on how to determine the ambient space of our code using supersingular Drinfeld modules of rank $R>1$. We start outlining the following result.

\begin{theorem}\label{thm: supersingular}
Let $\ell <q$, and consider $h=\prod_{i=1}^{\ell}(T-a_i)\in\vF_q[T]$ where $a_i\in \mathbb{F}^{\times}_q$ for every $i\in\{1,\ldots,\ell\}$. Let $\mathfrak{p} \in \mathbb{F}_q[T]$ be a monic irreducible polynomial of degree $m\geq \ell$, and $\alpha \in \mathbb{F}_{q^m}$ be one of its roots. Assume that $\phi: \mathbb{F}_q[T]\rightarrow \mathbb{F}_{q^m}\{\tau\}$ is a supersingular Drinfeld module of rank $R>1$ defined by $\phi_T=\alpha +\ldots+ g\tau^R$ for some $g\in \mathbb{F}^{\times}_{q^m}$.
Then, $\mathbb{F}_{q^m}(\phi[h])\subseteq \mathbb{F}_{q^{mRb}}$ if $b$ is the order of $\mathfrak{p}c\pmod h$ in the multiplicative group $\left(\faktor{\vF_q[T]}{\langle h \rangle}\right)^\times$ where $c=\frac{(-1)^{mR-m-R-1}}{\operatorname{Nr}_{{\mathbb{F}_{q^m}/\mathbb{F}_q}(g)}} \in \vF^{\times}_q$.
In particular, $\mathbb{F}_{q^m}(\phi[h])\subseteq \mathbb{F}_{q^{mR}}$ if $\mathfrak{p}c\equiv 1 \pmod h$.    
\end{theorem}
\begin{proof}
By \cite[Excercise 4.2.2]{Papikian2023}, we know that $\tau^{mR}=c\phi_{\mathfrak{p}}$. In other words, $\tau^{mR} = \phi_{\mathfrak{p}c}$. Taking into account the action on $\phi[h]$, we have $\tau^{mR}\vert_{\phi[h]}=\phi_{\mathfrak{p}c\pmod h}\vert_{\phi[h]}$. Hence, $\tau^{mRb}\vert_{\phi[h]}=1\vert_{\phi[h]}$ and $(\tau^{mRb}-1)\vert_{\phi[h]}=0\vert_{\phi[h]}$. In other words, we have that $\tau^{Rmb}$ acts like the identity on $\phi[h]$ which forces $\mathbb{F}_{q^m}(\phi[h])\subseteq\mathbb{F}_{q^{mRb}}$.     
\end{proof}

As a particular case of the last result, we show also the following.

\begin{corollary}
Let $\ell <q$, and consider $h=\prod_{i=1}^{\ell}(T-a_i)\in\vF_q[T]$ where $a_i\in \mathbb{F}^{\times}_q$ for every $i\in\{1,\ldots,\ell\}$. Let $\mathfrak{p} \in \mathbb{F}_q[T]$ be a monic irreducible polynomial of degree $m$, and $\alpha \in \mathbb{F}_{q^m}$ be one of its roots. Suppose $R$ is a prime such that $R\nmid m$, and consider the Drinfeld module $\phi: \mathbb{F}_q[T]\rightarrow \mathbb{F}_{q^m}\{\tau\}$ such that $\phi_T=\alpha+\tau^R$. Given $m<R\ell$, we have $\mathbb{F}_{q^m}(\phi[h])=\mathbb{F}_{q^{mR}}$ if and only if $\mathfrak{p}\equiv 1 \pmod h$.  
\end{corollary}
\begin{proof}
By Example~\ref{ex:supersing}, the Drinfeld module $\phi$ is indeed supersingular. Notice that $\#\phi[h]=q^{\ell R}$. If $m<\ell R$, then, $\mathbb{F}_{q^m}(\phi[h])\supsetneq \mathbb{F}_{q^m}$. In light of Theorem \ref{thm: supersingular}, when $\mathfrak{p}\equiv 1 \pmod h$, it follows that $\mathbb{F}_{q^m}(\phi[h])\subseteq \mathbb{F}_{q^{mR}}$. Since $R$ is a prime, it follows that $\mathbb{F}_{q^m}(\phi[h])= \mathbb{F}_{q^{mR}}$. Conversely, if $\mathbb{F}_{q^{mR}}$ is the $h$-division field of $\phi$ then $\phi_{\mathfrak{p}\pmod h}\vert_{\phi[h]} = \tau^{mR}\vert_{\phi[h]}= 1\vert_{\phi[h]}$ which implies that $\mathfrak{p}\equiv 1 \pmod h$. 
\end{proof}

We refer the reader to Section \ref{sec:exampleSuper} for an example of a code that is constructed based on the previous discussion on supersingular Drinfeld modules.

\subsection{Drinfeld modules of rank $R=1$}\label{sec: rank 1}
As shown in Section~\ref{sec:GeneralMainConstruction}, the rank $R$ of the Drinfeld module determines the size of the locality set of the code. This means that using directly rank one Drinfeld modules gives locality sets of size 1, and so it is not possible to recover any information. Nevertheless, if we use powers of Drinfeld modules of rank 1, the construction is still possible.

\begin{lemma}\label{thm: frob carlitz}
Suppose $R>1$. Let $\ell <q$, and consider $H=\prod_{i=1}^{\ell}(T^R-a_i)\in\vF_q[T]$ where $a_i\in \mathbb{F}^{\times}_q$ for every $i\in\{1,\ldots,\ell\}$. Consider a monic irreducible polynomial $\mathfrak{q} \in \mathbb{F}_q[T]$ of degree $m\geq\ell R$ and $\alpha \in \mathbb{F}_{q^m}$ one of its roots. Consider the Drinfeld module $\psi: \mathbb{F}_q[T]\rightarrow \mathbb{F}_{q^m}\{\tau\}$ defined by $\psi_T=\alpha + g\tau$ for some $g\in \mathbb{F}^{\times}_{q^m}$. Then, $\mathbb{F}_{q^m}(\psi[H])=\mathbb{F}_{q^{mb}}$ if and only if $b$ is the order of $\frac{\mathfrak{q}}{\operatorname{Nr}_{\mathbb{F}_{q^m}/\mathbb{F}_q}(g)}$ in the multiplicative group $\left(\faktor{\vF_q[T]}{\langle H \rangle}\right)^\times$. In particular, $\mathbb{F}_{q^m}(\psi[H])=\mathbb{F}_{q^m}$ if and only if $\mathfrak{q}\equiv \operatorname{Nr}_{\mathbb{F}_{q^m}/\mathbb{F}_q}(g) \pmod H$.  
\end{lemma}
\begin{proof}
Denote $\operatorname{Nr}_{\mathbb{F}_{q^m}/\mathbb{F}_q}(g)$ by $\operatorname{Nr}(g)$. Using \cite[Theorem 4.2.7(3)]{Papikian2023}, we know that  $\psi_{\mathfrak{q}}=\operatorname{Nr}(g)\tau^m$ which means $\psi_{\mathfrak{q}/\operatorname{Nr}(g)}=\tau^m$. Observe that $\mathbb{F}_{q^{mb}}$ is the $H$-division field of $\psi$ if and only if $ \tau^{mb}\vert_{\psi[H]}=1\vert_{\psi[H]}$ and $b$ is minimal, i.e. $\vF_{q^{mb}}$ is the smallest extension containing $\vF_{q^m}$ that contains $\phi[h]$. This is satisfied if and only if $ \psi_{\left(\mathfrak{q}/\operatorname{Nr}(g)\right)^b}\vert_{\psi[H]}=1\vert_{\psi[H]}$ (and $b$ is minimal) if and only if $ \left(\mathfrak{q}/\operatorname{Nr}(g)\right)^b \equiv 1 \pmod H$ (and $b$ is minimal). Finally, this is true if and only if $b$ is the order of $\mathfrak{q}/\operatorname{Nr}(g)$ in the multiplicative group $\left(\faktor{\vF_q[T]}{\langle H \rangle}\right)^\times$, as wanted.   
\end{proof}

Notice that Lemma \ref{thm: frob carlitz} can also be seen using  \cite[Theorem 7.1.19]{Papikian2023}. This lemma allows the construction of our code over a controllable extension field, as we now explain.
\begin{enumerate}

\item Choose $\ell$, $q$, $s$, $m$, and $R$ that satisfy conditions  C0, C1, C2 in Section~\ref{conditions} with $t=1$;

    \item Choose $a_i$'s $\in \mathbb{F}^\times_q$ and construct the polynomial $H=\prod_{i=1}^{\ell}(T^R-a_i)$;

 \item Find an irreducible polynomial $\mathfrak{q}\in \mathbb{F}_q[T]$ of degree $m\geq \ell R$ such that $\mathfrak{q}\equiv c \pmod H$, for some $c \in \mathbb{F}^{\times}_q$ (this is possible for most parameters using Dirichlet Theorem, see Section~\ref{sec:infinfam});

    \item Choose some $g\in \operatorname{Nr}^{-1}_{\mathbb{F}_{q^m}/\mathbb{F}_q}(c)$ and consider the Drinfeld module over $\mathbb{F}_{q^m}$ defined by $\psi_T=\alpha+g\tau$ where $\alpha$ is a root of $\mathfrak{q}$;
    
    \item By Lemma~\ref{thm: frob carlitz}, we have $\mathbb{F}_{q^m}(\psi[H])=\mathbb{F}_{q^{m}}\eqqcolon K$;
    
    \item Define the Drinfeld module $\phi$ of rank $R$ over $\mathbb{F}_{q^m}$ by $\phi_T=\psi_{T^R}$. Observe that the $A$-characteristic $\mathfrak{p}$ of $\phi$ is simply the minimal polynomial of $\alpha^R$;
    
    \item Notice that, since $\phi_{T-a_i}=\psi_{T^R-a_i}$, we have $\phi[T-a_i]=\psi[T^R-a_i]$ and therefore $\phi[h] = \phi\left[\prod_{i=1}^{\ell}(T-a_i)\right] = \psi[H]$. Notice also that $\gcd(\mathfrak{p},h)=1$;
    
    \item Use $\phi,\; h, \;\mathfrak{p}$ and $K$ to construct the code exactly as outlined in Section~\ref{sec:CodeConstruction}.
\end{enumerate}

For simplicity, in the above construction the congruence condition guarantees the nice inclusion $\phi[H]\subseteq \vF_{q^m}$, and therefore $t=1$, which allows a simpler treatment. By using the full power of Lemma \ref{thm: frob carlitz} it is possible to control completely the degree of the minimal extension where $\phi[H]$ is defined.

We show an example of this construction in  subsection~\ref{sec:exampleCarlitz}.

\section{Examples}\label{sec: examples}
In this section we illustrate our construction with two examples, using Carlitz modules and supersingular Drinfeld modules. For simplicity we directly work with $q$-linearized polynomials instead of polynomials in $\tau$.

\subsection{Explicit Construction via a Carlitz module}\label{sec:exampleCarlitz}
In this example we show how to construct our code using a Carlitz module. Let's start by choosing the following parameters: $q=5, r=2, \delta=2, \ell=3, s=2$. Thus, note that $R=r+\delta-1=3$. Take then $a_1=1,a_2=2,a_3=3$ and construct the polynomial $H\in\vF_5[T]$ as:
\[H = (T^3-1)(T^3-2)(T^3-3).\]
Now choose
\[\mathfrak{q} = T^{10} + 4T^9 + 4T^7 + T^6 + T^4 + 4T^3 + 4T + 2\in \vF_5[T].\]
A simple check shows that $\mathfrak{q}$ is irreducible and $\mathfrak{q}\equiv 1\pmod h$ (in particular it results $\mathfrak{q} = (T+4)h +1$). Let now $\vF_5(\alpha)\cong \vF_{5^{10}}$ be the splitting field of $\mathfrak{q}$ over $\vF_5$ where $\alpha$ is a primitive element for the extension (with minimal polynomial $x^{10} + 3x^5 + 3x^4 + 2x^3 + 4x^2 + x + 2$). Obviously $m=[\vF_5(\alpha):\vF_5]=10$ and notice that $\ell R = 9\nmid 10=m$, which implies that this code is not covered by the construction in \cite{rankmetriccodes}.

Choose now $z\in\vF_5(\alpha)$ to be a root of $\mathfrak{q}$, for example
\[z = 3\alpha^8 + 2\alpha^7 + \alpha^5 + 4\alpha^4 + 3\alpha^3 + 2\alpha^2 + 3\alpha + 2\]
and construct a (monic) Carlitz module over $\vF_5(\alpha)$ with $\vF_q[T]$-characteristic $\mathfrak{q}$:
\[C_T(x) = zx+x^5.\]
Therefore, we define $\phi_T$ as follows
\[
\phi_T(x)=C_{T^3}(x) = C_T(C_T(C_T(x))) =x^{5^3}+(z^{25}+z^5+z)x^{5^2}+(z^{10}+z^6+z^2)x^5+z^3x .
\]
As we have seen in Section~\ref{sec:TwoCases}, the code constructed via $H$, $\mathfrak{q}$ and $C_T$ is the same as the one constructed using $h$, $\mathfrak{p}$ and $\phi_T$, where $h=(T-1)(T-2)(T-3)$ and $\mathfrak{p}$ is the $A$-characteristic of $\phi$ (namely the minimal polynomial of $z^3$).

Now we can construct $W_1,W_2,W_3$. In this example we have that each $W_{a_i}$ is a vector sub-space of $\vF_5(\alpha)$, each one of dimension $R=3$ and $\vF_q$-generated by the roots of $C_{T^3}(x)-a_ix$, namely $\Ker(\phi_T(x)-a_ix)$, for $i\in \{1,2,3\}$. Thus, each of them is given by
\[
    W_1 = \langle \beta_1^1,\beta_2^1,\beta_3^1 \rangle \quad W_2 = \langle \beta_1^2,\beta_2^2,\beta_3^2 \rangle \quad W_3 = \langle \beta_1^3,\beta_2^3,\beta_3^3 \rangle
\]
where we can choose
\begin{align*}
    \beta_1^1 &= 4\alpha^8 + 2\alpha^7 + 2\alpha^5 + 2\alpha^4 + 3\alpha^2 + 1,\\
    \beta_2^1 &= \alpha^9 + 3\alpha^8 + 4\alpha^7 + 2\alpha^6 + 2\alpha^5 + 2\alpha^4 + 4\alpha^2 + \alpha,\\
    \beta_3^1 &= 4\alpha^9 + 3\alpha^8 + 4\alpha^7 + 2\alpha^5 + \alpha^3,\\
    \beta_1^2 &= 2\alpha^9 + 3\alpha^8 + 4\alpha^7 + \alpha^6 + 3\alpha^5 + \alpha^4 + 3\alpha^3 + 1,\\
    \beta_2^2 &= 4\alpha^8 + \alpha^7 + 4\alpha^6 + 2\alpha^5 + 4\alpha^4 + \alpha^3 + \alpha,\\
    \beta_3^2 &= 4\alpha^8 + 3\alpha^7 + \alpha^6 + 4\alpha^3 + \alpha^2,\\
    \beta_1^3 &= 4\alpha^7 + 3\alpha^6 + \alpha^5 + 2\alpha^4 + 2\alpha^3 + 1,\\
    \beta_2^3 &= 4\alpha^9 + \alpha^8 + 4\alpha^7 + 4\alpha^6 + 2\alpha^4 + \alpha,\\
    \beta_3^3 &= \alpha^9 + 3\alpha^7 + 2\alpha^6 + 2\alpha^5 + \alpha^4 + 2\alpha^3 + \alpha^2. 
\end{align*}
The message space $\mathcal M$ for our code is
\[
\mathcal M=\left\{\sum^{2}_{i=0} g_i(\phi_T^i(x)) \mid g_i\in \vF_5(\alpha)[x], \deg_q(g_i)\leq 1\right\}.
\]
Now take for example $f\in\mathcal M$ such that $g_0=x + \alpha x^5$, $g_1=(\alpha + 1)x^5$ and $g_2=(\alpha^2 + \alpha)x$, namely:
\[f= g_0(x) + g_1(\phi_{T^3}(x)) + g_2(\phi_{T^3}(\phi_{T^3}(x))).\]
Then $\enc(f)=[f_1,f_2,f_3,f_4,f_5,f_6,f_7,f_8,f_9]$ where
\begin{align*}
    f_1 &=f(\beta_1^1) = 2\alpha^9 + \alpha^8 + 2\alpha^7 + \alpha^6 + 3\alpha^3 + 3\alpha^2 + 4\alpha,\\
    f_2 &=f(\beta_2^1) = 4\alpha^9 + 2\alpha^8 + 2\alpha^6 + 4\alpha^5 + 2\alpha^4 + 3\alpha^3 + 3\alpha^2 + 3\alpha + 1, \\
    f_3 &=f(\beta_3^1) = 3\alpha^9 + 3\alpha^8 + \alpha^7 + 3\alpha^6 + 2\alpha^5 + 3\alpha^4 + 3\alpha^2 + 1,\\
    f_4 &=f(\beta_1^2) = 3\alpha^9 + \alpha^8 + 2\alpha^6 + \alpha^5 + \alpha^4 + 2\alpha^2 + 3\alpha + 1,\\
    f_5 &=f(\beta_2^2) = 3\alpha^7 + 3\alpha^6 + 3\alpha^5 + 2\alpha^4 + 4\alpha^2 + 4\alpha + 2,\\
    f_6 &=f(\beta_3^2) = 3\alpha^9 + 4\alpha^8 + 3\alpha^7 + 3\alpha^6 + 3\alpha^5 + 4\alpha^4 + 1,\\
    f_7 &=f(\beta_1^3) = 4\alpha^8 + 2\alpha^7 + 2\alpha^6 + 3\alpha^2 + \alpha + 3,\\
    f_8 &=f(\beta_2^3) = 2\alpha^8 + 3\alpha^7 + 2\alpha^6 + 2\alpha^4 + \alpha^3 + 2\alpha^2 + 3,\\
    f_9 &=f(\beta_3^3) = 2\alpha^9 + 2\alpha^8 + 2\alpha^7 + 2\alpha^6 + 4\alpha^5 + \alpha^4 + 4\alpha^3 + \alpha^2 + \alpha + 1.
\end{align*}

Suppose now, for example, that $f_5$ is erased and we want to recover it. Since $W_2$ is a locality set, we can use $f_4$ and $f_6$ to construct a \textit{recovery} polynomial $f|_{W_2}$ of $q$-degree at most $1$. Such polynomial is
\[f|_{W_2} = (3\alpha + 2)x^5 + (4\alpha^2 + 4\alpha + 1)x \]
and it is easy to verify that $f|_{W_2}(\beta_2^2)=f_5$.

\subsection{Explicit Construction via a supersingular Drinfeld module}\label{sec:exampleSuper}
This example shows how to construct our code using a supersingular Drinfeld module. Choose the following parameters: $q=5, r=2, \delta=2, \ell=3, s=2$. Thus $R=r+\delta-1=3$. Take then $a_1=1,a_2=2,a_3=3$ and construct the polynomial $h\in\vF_5[T]$ as:
\[h = (T-1)(T-2)(T-3).\]
Now choose $\mathfrak{p}\in\vF_5[T]$ such as
\[\mathfrak{p} = T^4 + T^3 + 4T^2 + T + 4.\]
A simple check shows that $\mathfrak{p}$ is irreducible and $\mathfrak{p}\equiv 1\pmod h$. Now, unlike the construction where Carlitz module is used, here the resulting code will be a rank-metric code over $\vF_{q^{Rm}}$ and once again, since $\ell R=9\nmid 12 = Rm$, this code is not covered by the construction in \cite{rankmetriccodes}.
Thus, consider $\vF_5(\alpha)\cong \vF_{5^{12}}$ where $\alpha$ is a primitive element for the extension (with minimal polynomial $x^{12} + x^7 + x^6 + 4x^4 + 4x^3 + 3x^2 + 2x + 2$).
Choose $z\in\vF_5(\alpha)$ to be a root of $\mathfrak{p}$, for example
\[z = 2\alpha^{10} + 2\alpha^9 + 3\alpha^8 + 2\alpha^7 + 4\alpha^4 + 3\alpha^2 + 4\alpha\]
and construct a supersingular Drifeld module over $\vF_5(\alpha)$ with $\vF_q[T]$-characteristic $\mathfrak{p}$, for example:
\[\phi_T(x) = zx+x^{5^3}.\]
To verify that such Drinfeld module is supersingular, just note that $3=R\nmid m=4$ (see \cite[Example 4.4.5]{Papikian2023}).
Now we can construct $W_1,W_2,W_3$. As before, we have that each $W_{a_i}$ is a vector sub-space of $\vF_5(\alpha)$, each one of dimension $R=3$ and $\vF_q$-generated by the roots of $\phi_{T}(x)-a_ix$, namely $\Ker(\phi_{T}(x)-a_ix)$, for $i\in \{1,2,3\}$. This means each of them is given by
\[
    W_1 = \langle \beta_1^1,\beta_2^1,\beta_3^1 \rangle \quad W_2 = \langle \beta_1^2,\beta_2^2,\beta_3^2 \rangle \quad W_3 = \langle \beta_1^3,\beta_2^3,\beta_3^3 \rangle
\]
where we can choose
\begin{align*}
    \beta_1^1 &= 3\alpha^{10} + 2\alpha^9 + 2\alpha^8 + 3\alpha^7 + \alpha^6 + 2\alpha^5 + 4\alpha^4 + 4\alpha^3 + 1,\\
    \beta_2^1 &= 3\alpha^{11} + 2\alpha^9 + 3\alpha^8 + 3\alpha^7 + 4\alpha^6 + 3\alpha^5 + 4\alpha^3 + \alpha,\\
    \beta_3^1 &= 2\alpha^8 + 3\alpha^7 + 2\alpha^6 + 2\alpha^4 + \alpha^3 + \alpha^2 ,\\
    \beta_1^2 &= 4\alpha^{11} + \alpha^{10} + 3\alpha^6 + \alpha^5 + 4\alpha^4 + \alpha^3 + 1 ,\\
    \beta_2^2 &= 4\alpha^{10} + 3\alpha^9 + 3\alpha^7 + 2\alpha^6 + \alpha^4 + \alpha^3 + \alpha,\\
    \beta_3^2 &= 3\alpha^{11} + 3\alpha^{10} + 4\alpha^8 + 4\alpha^7 + 4\alpha^6 + 2\alpha^5 + 4\alpha^4 + \alpha^3 + \alpha^2,\\
    \beta_1^3 &= 3\alpha^{11} + \alpha^{10} + \alpha^9 + 4\alpha^8 + 3\alpha^7 + 2\alpha^6 + 3\alpha^5 + 4\alpha^3 + 3\alpha^2 + 1,\\
    \beta_2^3 &= \alpha^{11} + \alpha^{10} + \alpha^9 + 2\alpha^8 + 2\alpha^7 + 3\alpha^5 + 2\alpha^3 + \alpha,\\
    \beta_3^3 &= 2\alpha^{11} + 2\alpha^{10} + 2\alpha^9 + \alpha^8 + 2\alpha^7 + 3\alpha^5 + \alpha^4. 
\end{align*}
The message space $\mathcal M$ for our code is
\[
\mathcal M=\left\{\sum^{2}_{i=0} g_i(\phi_T^i(x)) \mid g_i\in \vF_5(\alpha)[x], \deg_q(g_i)\leq 1\right\}.
\]
Now take for example $f\in\mathcal M$ such that $g_0=x + \alpha x^5$, $g_1=x^5$ and $g_2=\alpha^2 x$, namely:
\[f= g_0(x) + g_1(\phi_T(x)) + g_2(\phi_T(\phi_T(x))).\]
Then $\enc(f)=[f_1,f_2,f_3,f_4,f_5,f_6,f_7,f_8,f_9]$ where
\begin{align*}
    f_1 &=f(\beta_1^1) = 4\alpha^{11} + \alpha^{10} + 2\alpha^8 + 4\alpha^7 + 2\alpha^6 + 3\alpha^5 + 4\alpha^4 + 3\alpha^3 + 2\alpha^2 + 4 ,\\
    f_2 &=f(\beta_2^1) = 4\alpha^{11} + 4\alpha^{10} + 4\alpha^9 + 3\alpha^8 + 2\alpha^7 + 3\alpha^6 + 4\alpha^5 + 4\alpha^2 + 4\alpha , \\
    f_3 &=f(\beta_3^1) = 2\alpha^{10} + 4\alpha^9 + 4\alpha^8 + 3\alpha^7 + 2\alpha^5 + 2\alpha^4 + 2\alpha^3 + 4\alpha^2 + \alpha + 1 ,\\
    f_4 &=f(\beta_1^2) = 4\alpha^{11} + 2\alpha^{10} + 2\alpha^9 + 3\alpha^8 + 2\alpha^7 + \alpha^5 + 2\alpha^4 + 2\alpha^2 + \alpha,\\
    f_5 &=f(\beta_2^2) = 4\alpha^{11} + 2\alpha^9 + 4\alpha^8 + 4\alpha^7 + 3\alpha^6 + \alpha^5 + \alpha^4 + 4\alpha^2 + 4\alpha + 1,\\
    f_6 &=f(\beta_3^2) = 4\alpha^{11} + 4\alpha^{10} + 3\alpha^9 + \alpha^8 + 4\alpha^6 + 4\alpha^5 + 3\alpha^4 + \alpha^3 + 3\alpha^2 + 4,\\
    f_7 &=f(\beta_1^3) = 2\alpha^{11} + 3\alpha^{10} + 2\alpha^9 + \alpha^8 + \alpha^6 + \alpha^4 + \alpha^3 + \alpha^2 + 3\alpha + 4,\\
    f_8 &=f(\beta_2^3) = 2\alpha^{11} + \alpha^{10} + \alpha^9 + 2\alpha^8 + 2\alpha^7 + 4\alpha^4 + \alpha^3 + 2\alpha^2 + 2,\\
    f_9 &=f(\beta_3^3) =2\alpha^{11} + \alpha^9 + 3\alpha^7 + 3\alpha^6 + 2\alpha^5 + 4\alpha^4 + 4\alpha^3 + 4\alpha^2 .
\end{align*}

As in the previous example, suppose that $f_5$ is erased and we want to recover it. Since $W_2$ is a locality set, we can use $f_4$ and $f_6$ to construct a \textit{recovery} polynomial $f|_{W_2}$ of $q$-degree at most $1$. Such polynomial is
\[f|_{W_2} = (\alpha + 2)x^5 + (4\alpha^2 + 1)x \]
and it is easy to verify that $f|_{W_2}(\beta_2^2)=f_5$.

\subsection{Choosing the polynomial $\mathfrak{p}$}\label{sec:ExampleOfP}
In these examples, a fundamental role is played by the polynomial $\mathfrak{p}$ (and the attached Drinfeld module for which $\vF_{q^m}(\phi[h])\subseteq K = \vF_{q^{tm}}$) that, a priori, might not exist. In Section \ref{sec:infinfam} we show an infinite family of parameters for which $\mathfrak{p}$ (or $\mathfrak{q}$ in the case of rank-1 Drinfeld modules) exists. However, such $\mathfrak{p}$ can be found also for many other choices of parameters that are not necessarily prescribed by Section \ref{sec:infinfam}. For example, in the Table~\ref{tablePs} we have reported some choices for $\mathfrak{q}=uh+1$ ($K=\vF_{q^m}$), that allow to construct our code using a Carlitz module.

\renewcommand{\arraystretch}{1.5} 
\counterwithout{table}{section}
\begin{table}[!h]
\caption{}
\label{tablePs}
\begin{center}
\addtolength{\leftskip} {-2.15cm}
\begin{tabular}{ |>{\centering\arraybackslash}m{.3cm} | >{\centering\arraybackslash}m{.3cm} | >{\centering\arraybackslash}m{.3cm} | >{\centering\arraybackslash}m{3.5cm} | >{\centering\arraybackslash}m{.3cm} | >{\centering\arraybackslash}m{2cm} | >{\centering\arraybackslash}m{8cm}| }
	\hline
 	$q$ & $R$ & $\ell$ & $H=\prod_{i=1}^\ell(T^R-a_i)$ & $m$ & $u$ & $\mathfrak{q}=uh+1$ \\
    
    \hline\hline
	\multirow{5}{*}{5} & \multirow{5}{*}{3} & \multirow{5}{*}{2} & \multirow{4}{*}{$T^6 + 2T^3 + 2$} & 6 & $1$ & $T^6 + 2T^3 + 3$ \\
	\cline{5-7}
 	&&&& 7 & $T + 3$ & $T^7 + 3T^6 + 2T^4 + T^3 + 2T + 2$ \\
 	\cline{5-7}
 	&&&& 8 & $T^2 + 4$ & $T^8 + 4T^6 + 2T^5 + 3T^3 + 2T^2 + 4$ \\
 	\cline{5-7}
 	&&& $(a_1,a_2)=(1,2)$ & 9 & $T^3 + T + 3$ & $T^9 + T^7 + 2T^4 + 3T^3 + 2T + 2$ \\
 	\cline{5-7}
 	&&&& 10 & $T^4 + 4$ & $T^10 + 2T^7 + 4T^6 + 2T^4 + 3T^3 + 4$ \\
    
    \hline\hline
    
	\multirow{5}{*}{5} & \multirow{5}{*}{3} & \multirow{5}{*}{3} & \multirow{4}{*}{$T^9 + 4T^6 + T^3 + 4$} & 10 & $T + 2$ & $T^{10} + 2T^9 + 4T^7 + 3T^6 + T^4 + 2T^3 + 4T + 4$ \\
	\cline{5-7}
 	&&&& 11 & $T^2 + T$ & $T^{11} + T^{10} + 4T^8 + 4T^7 + T^5 + T^4 + 4T^2 + 4T + 1$ \\
 	\cline{5-7}
 	&&&& 12 & $T^3 + 2$ & $T^{12} + T^9 + 4T^6 + T^3 + 4$ \\
 	\cline{5-7}
 	&&& $(a_1,a_2,a_3)=(1,2,3)$ & 13 & $T^4 + 2T + 4$ & $T^{13} + T^{10} + 4T^9 + 4T^7 + T^6 + T^4 + 4T^3 + 3T + 2$ \\
 	\cline{5-7}
 	&&&& 14 & $T^5 + 4T$ & $T^{14} + 4T^{11} + 4T^{10} + T^8 + T^7 + 4T^5 + 4T^4 + T + 1$ \\
 	
 	\hline\hline
 	 
 	\multirow{5}{*}{7} & \multirow{5}{*}{3} & \multirow{5}{*}{2} & \multirow{4}{*}{$T^6 + 4T^3 + 2$} & 7 & $T$ & $T^7 + 4T^4 + 2T + 1$ \\
	\cline{5-7}
 	&&&& 8 & $T^2 + 1$ & $T^8 + T^6 + 4T^5 + 4T^3 + 2T^2 + 3$ \\
 	\cline{5-7}
 	&&&& 9 & $T^3 + 1$ & $T^9 + 5T^6 + 6T^3 + 3$ \\
 	\cline{5-7}
 	&&& $(a_1,a_2)=(1,2)$ & 10 & $T^4 + T + 5$ & $T^{10} + 5T^7 + 5T^6 + 6T^4 + 6T^3 + 2T + 4$ \\
 	\cline{5-7}
 	&&&& 11 & $T^5+6$ & $T^{11} + 4T^8 + 6T^6 + 2T^5 + 3T^3 + 6$ \\
 	
 	\hline\hline
 	
 	\multirow{5}{*}{8} & \multirow{5}{*}{3} & \multirow{5}{*}{2} & \multirow{4}{*}{$T^6 + (\alpha + 1)T^3 + \alpha$} & 7 & $T + \alpha^2 + \alpha$ & $T^7 + (\alpha^2 + \alpha)T^6 + (\alpha + 1)T^4 + T^3 + \alpha T + \alpha^2 + \alpha$ \\
	\cline{5-7}
 	&&&& 8 & $T^2$ & $T^8 + (\alpha + 1)T^5 + \alpha T^2 + 1$ \\
 	\cline{5-7}
 	&&&& 9 & $T^3 + \alpha T$ & $T^9 + \alpha T^7 + (\alpha + 1)T^6 + (\alpha^2 + \alpha)T^4 + \alpha T^3 + \alpha^2T + 1$ \\
 	\cline{5-7}
 	&&& $(a_1,a_2)=(1,\alpha)$ & 10 & $T^4$ & $T^{10} + (\alpha + 1)T^7 + \alpha T^4 + 1$ \\
 	\cline{5-7}
 	&&&& 11 & $T^5$ & $T^{11} + (\alpha + 1)T^8 + \alpha T^5 + 1$ \\
 	
 	\hline
 	
\end{tabular}
\end{center}
\end{table}

\section{An Infinite Family of Optimal Rank Metric Codes with Locality}\label{sec:infinfam}

In this section we will prove Theorem \ref{thm:GeneralMainThm} in the case of rank-1 and supersingular Drinfeld modules as a corollary of Section~\ref{sec:CodeConstruction}, Section~\ref{sec:TwoCases} and Dirichlet Theorem for polynomial arithmetic progressions.
We follow Rosen's  line of work \cite{Rosen2002} to extract the information needed to guarantee the existence of a polynomial $\mathfrak{p}$ of the wanted form.
One of our tasks will be to precisely compute some error terms in \cite[Chapter 4]{Rosen2002}  (which is interesting on its own right, since effective versions of these results are difficult to find in the literature).

Throughout this section let $\mathbb C$ be the field of complex numbers, $q$ be a prime power, $A=\vF_q[T]$ be the ring of polynomials over the finite field with $q$ elements, and $h$ be a fixed element of $A$ of positive degree. Unless otherwise specified, we write $P$ to denote a monic irreducible polynomial in $A$, so a sum whose index is $P$ is meant to run only over irreducible polynomials.

\subsection{$L$-series and Dirichlet Theorem}
\begin{theorem}\label{appx:cardinalitySm}
    Let $S_m$ denote the set of monic irreducible polynomials in $A$ of degree $m$. Then, \[ \left| \#S_m-\frac{q^m}{m}\right|\leq 2\frac{q^{\frac{m}{2}}}{m}. \]
\end{theorem}
\begin{proof}
    It follows directly from \cite[Theorem 2.2]{Rosen2002}.
\end{proof}

\begin{definition}\label{def:character}
    A \emph{Dirichlet character} modulo $h$ is a function $\chi:A\rightarrow \mathbb{C}$ such that:
    \begin{itemize}
        \item[(a)] $\chi(a+bh)=\chi(a)$ for all $a,b\in A$
        \item[(b)] $\chi(a)\chi(b)=\chi(ab)$ for all $a,b\in A$
        \item[(c)] $\chi(a)\neq 0$ if and only if $\gcd(a,h)=1$
    \end{itemize}
\end{definition}

From part (b), it follows immediately that $\chi(1)=1$. Generally, it can be shown that the value of a character is either zero or a root of unity, and so $|\chi(a)|\in\{0,1\}$ for every $a\in A$.

Let $\Phi(h)$ be the number of non-zero polynomials of degree less than $\deg(h)$ and relatively prime to $h$. It can be shown that there are exactly $\Phi(h)$ Dirichlet characters modulo $h$ (see \cite[Chapter 4]{Rosen2002}).
\begin{definition}\label{chitrivial}
  For every $a\in A$, the \emph{trivial} Dirichlet character modulo $h$ is denoted by $\chi_o$ and defined by $\chi_o(a)=1$ if $\gcd(a,h)=1$, $\chi_o(a)=0$ if $\gcd(a,h)\neq 1$. 
\end{definition}
Let $X_h$ denote the set of all the Dirichlet characters modulo $h$.  For every $\chi,\psi\in X_h$, we can define the product by $\chi\psi(a) \coloneqq \chi(a)\psi(a)$ and the inverse by $\chi^{-1}(a)\coloneqq \chi(a)^{-1}$. With these operations, $X_h$ can be seen as a group with identity element $\chi_o$. Finally, since characters are complex-value functions, we write $\overline{\chi(a)}$ or $\bar{\chi}(a)$ for the complex conjugate of $\chi(a)$. Next, we outline a result about orthogonality relations between characters. 

\begin{proposition}\cite[Proposition 4.2]{Rosen2002}\label{appx:orthogonalityr}
    Let $\chi,\psi$ be Dirichlet characters modulo $h$ and $a,b\in A$ that are coprime to $h$. Let $\delta$ be the Kronecker delta. Then, 
    \begin{itemize}
        \item [(1)] $\sum_a\chi(a)\overline{\psi(a)}=\Phi(h)\delta(\phi,\psi)$
        \item [(2)] $\sum_\chi\chi(a)\overline{\chi(b)}=\Phi(h)\delta(a,b)$
    \end{itemize}
    The  sum in (1) is over any set of representative for $\faktor{A}{hA}$ and the sum in (2) is over all Dirichlet characters modulo $h$.
\end{proposition}

From now on, for every $f\in A$, let $|f|=q^{\deg(f)}$  and $s$ be a complex variable with real part greater than 1.

\begin{definition}\label{def:lserieschi}
    Let $\chi$ be a Dirichlet character modulo $h$. The \emph{Dirichlet L-series} corresponding to $\chi$ is defined by \[L(s,\chi) = \sum_{f\text{ monic}}\frac{\chi(f)}{|f|^s}.\]
\end{definition}

\begin{proposition}[Euler product]\label{appx:eulerproduct}
    The following equality holds
    \begin{equation}\label{eq:seriesX}
        L(s,\chi) = \prod_{P}\left( 1-\frac{\chi(P)}{|P|^s} \right)^{-1}.
    \end{equation}
    In particular, for the $L$-series corresponding to the trivial character we have
    \begin{equation}\label{eq:seriesX0}
        L(s,\chi_o) = \prod_{P|h}\left( 1-\frac{1}{|P|^s}\right)\zeta_A(s)
    \end{equation}
    where 
    \begin{equation}\label{eq:zeta}
        \zeta_A(s) = \sum_{\substack{f\in A \\ f\text{ monic}}}\frac{1}{|f|^s} = \frac{1}{1-q^{1-s}}
    \end{equation}
    is the Riemann Zeta function for $A$.
\end{proposition}
\begin{proof}
    See \cite[Chapter 4]{Rosen2002}. The proof follows from the definition of L-series and comparison with the properties of $\zeta_A(s)$, written in the product form \[\zeta_A(s)=\prod_P\left( 1-\frac{1}{|P|^s} \right)^{-1}.\]
\end{proof}
Notice that, from part (c) of Definition \ref{def:character}, we can rearrange Equation~\eqref{eq:seriesX} as follows:
\begin{align}\label{eq:pnoth}
    \begin{split}
        L(s,\chi) &= \prod_{P}\left( 1-\frac{\chi(P)}{|P|^s} \right)^{-1} = \prod_{P\nmid h}\left( 1-\frac{\chi(P)}{|P|^s} \right)^{-1}\prod_{P|h}\left( 1-\frac{\chi(P)}{|P|^s} \right)^{-1} \\
        &= \prod_{P\nmid h}\left( 1-\frac{\chi(P)}{|P|^s} \right)^{-1}.
    \end{split}
\end{align}

\begin{proposition}\cite[Propostion 4.3]{Rosen2002}\label{appx:seriesproduct}
    Let $\chi$ be a non-trivial Dirichlet character modulo $h$. Then, $L(s,\chi)$ is a polynomial in $q^{-s}$ of degree at most $\deg(h)-1$.
\end{proposition}

In what follows, we will need the result below. To prove it, we adjust the proof of \cite[Theorem 4.8]{Rosen2002} to our special case in order to extract explicit constants.

\begin{theorem}[Explicit Dirichlet Theorem]\label{thm:Dirichlet}
Let $h\in A$ and let $S_m$ be the set of monic irreducible polynomials in $A$ of degree $m$, $T(m)$ be the number of divisors of $m$. For $a\in \left(\faktor{A}{hA}\right)^\times$, let $S_m(a,h):=\{P\in S_m: P\equiv a\mod h\}$. Then \[ \left| \#S_m(a,h)-\frac{q^m}{m\Phi(h)}\right| \leq \left(3T(m)+2\deg(h)\right)\frac{q^{\frac{m}{2}}}{m}.\]
\end{theorem}
\begin{proof}
Set $H=\deg(h)$ and consider the $L$-series $L(s,\chi)$. We want to express it as a product in two different ways and compare the coefficients. By Proposition \ref{appx:seriesproduct}, we know that $L(s,\chi)$ is a polynomial in $q^{-s}$ of degree at most $H-1$. Therefore, setting $u=q^{-s}$ we can write
\begin{equation}\label{eq:L1}
    L^*(u,\chi)\coloneqq \sum_{k=0}^{H-1}a_k(\chi)u^k = \prod_{i=1}^{H-1}(1-\alpha_i(\chi)u)
\end{equation}
where $\alpha_i(\chi)$ is the inverse of a complex root of $L^*(u,\chi)$ for every $i\in\{1,\ldots, H-1\}$. The second expression for $L^*(u,\chi)$ comes from rewriting the Euler product for $L(s,\chi)$ in terms of $u$. Using Equation \eqref{eq:pnoth} and reordering the factors we have
\[L(s,\chi) = \prod_{P\nmid h}\left( 1-\frac{\chi(P)}{|P|^s} \right)^{-1} = \prod_{d=1}^\infty \prod_{\substack{P\nmid h \\ \deg(P)=d}}(1-\chi(P)q^{-ds})^{-1}\] and setting again $u=q^{-s}$ we get
\begin{equation}\label{eq:L2}
    L^*(u,\chi) = \prod_{d=1}^\infty \prod_{\substack{P\nmid h \\ \deg(P)=d}}(1-\chi(P)u^d)^{-1}.
\end{equation}
The idea of the proof is to take the logarithmic derivative of both sides and compare the coefficients. We now state an identity for the logarithmic derivative which will be used frequently.
\begin{equation}\label{eq:log}
    u\frac{d}{du}(\log(1-\alpha u)^{-1}) = \sum_{m=1}^\infty \alpha^m u^m
\end{equation}
where $\alpha$ is a complex number. The proof is trivial using geometric series for $|u|<|\alpha|^{-1}$. Now, for each character $\chi$ modulo $h$ we define the number $c_m(\chi)$ such that \[u\frac{d}{du}\log(L^*(u,\chi)) = \sum_{m=1}^\infty c_m(\chi)u^m.\] We have two cases: $\chi=\chi_o$ and $\chi\neq\chi_o$.

First, from Equation \eqref{eq:seriesX0} and \eqref{eq:zeta}, writing $u=q^{-s}$ as always, we get \[L^*(u,\chi_o) = \prod_{P|h}\left(1-u^{\deg(P)}\right)\frac{1}{1-qu}.\] Now, consider the inverse of the last equation. We take the logarithmic derivative and apply its additive property, we use Equation \eqref{eq:log} and get
\begin{align*}
    u\frac{d}{du}\log(L^*(u,\chi_o)) &= -u\frac{d}{du}\log(L^*(u,\chi_o)^{-1}) \\
    &=\sum_{m=1}^\infty q^mu^m - \sum_{m=1}^\infty\sum_{P|h}u^{m\deg(P)}\\
    &=\sum_{m=1}^\infty\left(q^mu^m-\sum_{P|h}u^{m\deg(P)}\right).
\end{align*}

Therefore, since in the worst case scenario for some $m$ the inner sum contributes with $-Hu^m$, comparing with the definition of $c_m(\chi)$, it is now obvious that 
\begin{equation}\label{eq:cn0}
   |c_m(\chi_o)-q^m|\leq H. 
\end{equation}

Instead, for $\chi\neq\chi_o$, take the logarithmic derivative of the inverse of $L^*(u,\chi)$ using Equation \eqref{eq:L1}, multiply by $u$, and combine the resulting equation with Equation \eqref{eq:log}. We have 

\[c_m(\chi) = -\sum_{k=1}^{H-1}\alpha_k(\chi)^m.\]

From the analogue of the Riemann hypothesis for function fields over a finite field, it follows that each of the $\alpha_k(\chi)$ has absolute value either 1 or $\sqrt{q}$ (see \cite{weil1948courbes}). Therefore
\begin{equation}\label{eq:cngeneric}
    |c_m(\chi)|\leq (H-1)q^{\frac{m}{2}}.
\end{equation}
Consider now the Euler product expansion of $L^*(u,\chi)$ given by Equation \eqref{eq:L2}. Take the logarithmic derivative, multiply both sides of the resulting equation by $u$ and again, using Equation \eqref{eq:log}, we get \[c_m(\chi)=\sum_{\substack{k,P \\ k\deg(P)=m}}\deg(P)\chi(P)^k\] that can be rearranged by separating the term corresponding to $k=1$ as outlined below.
\[c_m(\chi) = m\sum_{\deg(P)=m}\chi(P) + \sum_{\substack{d|m \\ d\leq m/2}}d\sum_{\deg(P)=d}\chi(P)^{m/d}.\]
Since the value of a character is either zero or a root of unity, the absolute value of the inner sum on the right is less than or equal to the number of monic irreducible polynomials of degree $d$. Thus, using Theorem \ref{appx:cardinalitySm} we have
\begin{equation}\label{eq:cmpre}
    \left| c_m(\chi)-m\sum_{\deg(P)=m}\chi(P)\right| \leq \sum_{\substack{d|m \\ d\leq m/2}}d\left( \frac{q^d}{d}+2\frac{q^\frac{d}{2}}{d}\right)\leq 3q^{\frac{m}{2}}T(m)
\end{equation}
where $T(m)$ is the number of divisors of $m$.

To conclude the proof, we need now to compute the expression $\sum_{\chi}\bar{\chi}(a)c_m(\chi)$ in two ways. Remember that since $a\in \left(\faktor{A}{hA}\right)^\times$, we have that $\gcd(a,h)=1$.  We have

\begin{equation*} 
    \begin{split}
        &\left| \frac{1}{\Phi(h)}\sum_\chi\bar{\chi}(a)c_m(\chi)-m\#S_m(a,h) \right|= \\
        & =\left|\frac{1}{\Phi(h)}\sum_\chi\bar{\chi}(a)\left( c_m(\chi)-m\sum_{\deg(P)=m}\chi(P) + m\sum_{\deg(P)=m}\chi(P) \right) -m\#S_m(a,h)\right|
		\\        
        &\leq U+V
    \end{split}
\end{equation*}
where \[U=\left|\frac{1}{\Phi(h)}\sum_\chi\bar{\chi}(a)\left( c_m(\chi)-m\sum_{\deg(P)=m}\chi(P)\right)\right|\]
and 
\[V=\left| \frac{m}{\Phi(h)}\sum_{\deg(P)=m}\sum_\chi\bar{\chi}(a)\chi(P)  -m\#S_m(a,h)\right|.\]

Finally, using Equation  \eqref{eq:cmpre} for $U$, and the orthogonality relations from Proposition \ref{appx:orthogonalityr} for $V$, we can upper bound $U+V$ as follows:
\begin{equation*}
\begin{split}
        &U+V\leq \frac{3q^{\frac{m}{2}}}{\Phi(h)}T(m)\sum_\chi\left|\bar{\chi}(a)\right| + \left|m\sum_{\deg(P)=m}\delta(a,P)-m\#S_m(a,h)\right|\\
        &=  \frac{3q^{\frac{m}{2}}}{\Phi(h)}T(m)\Phi(h) = 3q^{\frac{m}{2}}T(m)
    \end{split}
\end{equation*}
namely
\begin{equation}\label{final1}
    \left| \sum_\chi\bar{\chi}(a)c_m(\chi)-m\Phi(h)\#S_m(a,h) \right|\leq 3q^{\frac{m}{2}}T(m)\Phi(h).
\end{equation}

Next, using Equations \eqref{eq:cn0} and \eqref{eq:cngeneric}, that $|\chi(a)|=1$ and $\chi_o(a)=1$ (since $\gcd(a,h)=1$), and the triangle inequality we get

\begin{align}\label{final2}
    \begin{split}
        \left| \sum_\chi\bar{\chi}(a)c_m(\chi)-q^m \right| & = \left|\bar{\chi_o}(a)c_m(\chi_o)-q^m + \sum_{\chi\neq\chi_o}\bar{\chi}(a)c_m(\chi) \right| \\
        &\leq H+ \sum_{\chi\neq\chi_o}(H-1)q^{\frac{m}{2}}\leq 2Hq^{\frac{m}{2}}\Phi(h).
    \end{split}
\end{align}

Finally, using the last two results \eqref{final1} and \eqref{final2} we have
\begin{equation*}
    \begin{split}
        \left|m\Phi(h)\#S_m(a,h)-q^m \right| =&\left|m\Phi(h)\#S_m(a,h)-\sum_\chi\bar{\chi}(a)c_m(\chi)+\sum_\chi\bar{\chi}(a)c_m(\chi)-q^m \right| \\
        \leq & \left| \sum_\chi\bar{\chi}(a)c_m(\chi)-m\Phi(h)\#S_m(a,h) \right|+\left| \sum_\chi\bar{\chi}(a)c_m(\chi)-q^m \right| \\
        \leq & 3q^{\frac{m}{2}}T(m)\Phi(h) + 2\deg(h)q^{\frac{m}{2}}\Phi(h).
    \end{split}
\end{equation*}
Dividing by $m\Phi(h)$ concludes the proof. 
\end{proof}

\subsection{Existence of an infinite family}

Finally, we are able to employ Dirichlet Theorem to guarantee the existence of an infinite family of our codes. First, we need a preliminary lemma.

\begin{lemma}\label{prop:selectprime}
Set $H=\prod^{\ell}_{i=1} (T^R-a_i)$.
Let $m\geq 4$ be an integer such that \[\ell R<\frac{m}{2}-\frac{\log(m)}{\log(q)}-\frac{\log\left(5 \right)}{\log(q)}.\] Then there exists a monic irreducible polynomial $\mathfrak q \in \vF_q[T]$ of degree $m$ such that $\mathfrak q \equiv 1\pmod h$.
\end{lemma}
\begin{proof}
    For $a=1$ and  in Theorem \ref{thm:Dirichlet} we know that \[ \left| \#S_m(1,h)-\frac{q^m}{m\Phi(h)}\right| \leq 3\frac{q^{\frac{m}{2}}}{m}T(m)+2\deg(h)\frac{q^{\frac{m}{2}}}{m}.\]

    Since we can bound $T(m)\leq2\sqrt{m}$ and $\deg(h)=\ell R\leq m$,  we can write 
    \[  \#S_m(1,h)\geq \frac{q^m}{m\Phi(h)} - 6\frac{q^{\frac{m}{2}}}{\sqrt{m}}-2q^{\frac{m}{2}}.\]
    
Bounding now $\Phi(h)$ simply with $q^{\ell R}$ and observing that we want to estimate $m$ for which $\#S_m(1,h)$ is at least $1$, it is sufficient to check when \[0<\frac{q^m}{mq^{\ell R}}-6\frac{q^{\frac{m}{2}}}{\sqrt{m}}-2q^{\frac{m}{2}}.\]
This gives 
    \[\left(\frac{6}{\sqrt{m}}+2\right)< \frac{q^{\frac{m}{2}}}{mq^{\ell R}}.\]
    Taking the logarithm on both sides concludes the proof \[\ell R<\frac{m}{2}-\frac{\log(m)}{\log(q)}-\frac{\log\left( 6/\sqrt{m}+2 \right)}{\log(q)},\]
    since this condition is verified by hypothesis for $m\geq 4$.
\end{proof}

Finally, we are ready to prove the existence of an infinite family of our codes.

\begin{theorem}\label{thm:Existence}
Let $q$ be a prime power, and let $m\geq 4$ be an integer. Let $R,\ell$ be positive integers such that \[\ell R<\frac{m}{2}-\frac{\log(m)}{\log(q)}-\frac{\log\left( 5 \right)}{\log(q)},\]
$2\leq \delta< R$, $s+1\leq \ell< q$. Then, there exists a  rank-metric code in $\vF_q^{m\times \ell R}$, rank-locality $(r,\delta)$, $R=r+\delta-1$, and $\vF_q$-dimension equal to $m(s+1)r$. 
Moreover, this code has rank distance $\ell R-Rs-r+1$, thus it is optimal with respect to the bound in Proposition \ref{prop: singletonLikeRankMetric}.
\end{theorem}
\begin{proof}
Employ the construction of subsection \ref{sec: rank 1} observing that with the hypothesis of the Theorem the parameters verify conditions C0, C1, C2 (and therefore (1) of subsection \ref{sec: rank 1}).
Now, since we have
\[\ell R<\frac{m}{2}-\frac{\log(m)}{\log(q)}-\frac{\log\left( 5 \right)}{\log(q)},\]
Lemma \ref{prop:selectprime} guarantees the existence of an irreducible polynomial $\mathfrak q$ of degree $m$ that is congruent to $1$ modulo $H$, which verifies (3) in the construction of subsection \ref{sec: rank 1}. 
This allows for the construction of subsection \ref{sec: rank 1} to happen with $g=1$, i.e. for the rank $1$ Drinfeld module $T\mapsto \alpha+\tau$, with $\alpha$ root of $\mathfrak q$ in $\vF_{q^m}$. 

 Now derive the wanted parameters from Theorem \ref{thm:GeneralMainThm}.
\end{proof}

\bibliography{ref.bib}
\bibliographystyle{abbrv}

\end{document}